\documentclass[conference]{IEEEtran}
\usepackage{subfigure}
\usepackage{amsmath}
\usepackage{graphicx,amssymb,lineno,bm,subfigure,color}
\usepackage{algorithm}
\usepackage{algorithmic}
\renewcommand{\baselinestretch}{1}

\newtheorem{proof}{\textbf{Proof}}
\newtheorem{lma}{\textbf{Lemma}}

\renewcommand{\baselinestretch}{0.986}

\ifCLASSINFOpdf
\else
\fi
\IEEEoverridecommandlockouts
\begin{document}
\title{Selective Uplink Training for Massive MIMO Systems}
\author{\IEEEauthorblockN{Changming Li$^*$, Jun Zhang$^*$, Shenghui Song$^*$, and K. B. Letaief$^*$$^\dag$, \emph{Fellow, IEEE} }
\IEEEauthorblockA{$^*$Dept. of ECE, The Hong Kong University of Science and Technology, $^\dag$Hamad Bin Khalifa University, Doha, Qatar\\
Email: $^*$\{cliao, eejzhang, eeshsong, eekhaled\}@ust.hk, $^\dag$kletaief@hbku.edu.qa}
\thanks{This work is partially supported by the Hong Kong Research Grants Council under Grant No. 16211815.}
}



%


\maketitle

\begin{abstract}
As a promising technique to meet the drastically growing demand for both high throughput and uniform coverage in the fifth generation (5G) wireless networks, massive multiple-input multiple-output (MIMO) systems have attracted significant attention in recent years. However, in massive MIMO systems, as the density of mobile users (MUs) increases, conventional uplink training methods will incur prohibitively high training overhead, which is proportional to the number of MUs. In this paper, we propose a selective uplink training method for massive MIMO systems, where in each channel block only part of the MUs will send uplink pilots for channel training, and the channel states of the remaining MUs are predicted from the estimates in previous blocks, taking advantage of the channels' temporal correlation. We propose an efficient algorithm to dynamically select the MUs to be trained within each block and determine the optimal uplink training length. Simulation results show that the proposed training method provides significant throughput gains compared to the existing methods, while much lower estimation complexity is achieved. It is observed that the throughput gain becomes higher as the MU density increases.
\end{abstract}

%
\IEEEpeerreviewmaketitle

\section{Introduction}
With the advances of the Internet of Things (IoT) and Machine-to-Machine (M2M) communications, wireless data traffic is witnessing an unprecedented growth. In order to provide seamless wireless access and to achieve satisfactory quality of service (QoS), massive multiple-input multiple-output (MIMO) has recently emerged as a promising technology for the next generation wireless networks. By equipping base stations (BSs) with a large number of antennas, massive MIMO systems bring various attractions such as higher system throughput and energy-efficiency \cite{massive_mimo}.

To fully exploit the benefits of massive MIMO systems, transmission protocols, such as interference management and resource allocation strategies, should be carefully designed, in which the channel side information (CSI) plays a critical role. It has been demonstrated that the achievable performance of massive MIMO systems is closely related to the quality of the available CSI \cite{massive_mimo}. However, obtaining the high-dimensional CSI in massive MIMO systems requires a substantial amount of training and feedback overhead, thanks to the large number of antennas at BSs. Consequently, time division duplexing (TDD) massive MIMO has emerged as an attractive candidate, for which the CSI for both the uplink and downlink is obtained via uplink training, with the overhead proportional to the number of mobile users (MUs) \cite{inter_cell_interference}.

However, as the density of mobile devices keeps increasing, the uplink training overhead will grow proportionally, which will limit the spectrum efficiency of massive MIMO systems \cite{massive_fundamentals}. Therefore, innovative methodologies for training overhead reduction will be needed. One such method is to adopt non-orthogonal pilots for channel training, e.g., \emph{Gaussian random} sequences, \emph{generalized Welch bound equality} sequences \cite{user_capacity}, or \emph{Grassmannian subspace packing} sequences \cite{downlink_training}. However, non-orthogonal pilots are generally difficult to design and the performance characterization is typically intractable. This has motivated the development of alternative methods to reduce training overhead, while retaining the simple-to-implement orthogonal pilots. One way to achieve this is to train and transmit to a subset of MUs during each block, which can be achieved via user scheduling, e.g., \emph{Round-Robin Scheduling} (RRS) or \emph{Priority-based Scheduling} (PS) \cite{multiuser_MIMO, achieving_massive}. In this way, the training overhead is reduced to be proportional to the number of the MUs in the subset. However, this comes at the expense of lower spectral efficiency, as the MUs are served in a time division multiple access (TDMA) manner.

To effectively reduce training overhead and improve spectral efficiency for massive MIMO systems, it is critical to exploit the unique structures of massive MIMO channels, such as the sparse structure in the angular domain \cite{sparsity, sparsity2} and antenna correlation \cite{jointspatial}. In this paper, we exploit another key characteristic, i.e., the temporal correlation of the channel. In most of the existing works, simplified channel models, e.g., the independent and identically distributed (i.i.d.) block fading channel model, are assumed for ease of analysis \cite{user_capacity,how_much,optimal_channel}. However, these models cannot capture the channel's temporal correlation, which exists especially in low-mobility environments. There are some recent works applying Kalman filter-based training methods to exploit the channel's temporal correlation \cite{downlink_training}, which, however, suffer from high computational complexity.

In this paper, we investigate the uplink training in TDD massive MIMO systems, and propose a selective training method which effectively reduces training overhead and significantly improves spectrum efficiency. In each channel block, the BS selects part of the MUs for uplink training, while the CSI of the remaining MUs is obtained by prediction based on the estimates in previous blocks, exploiting the temporal correlation. In the data transmission phase, the BS serves all the MUs simultaneously with the obtained CSI, either from channel training or prediction. Thus the proposed method enjoys much lower training overhead compared to the full training case, and also much lower estimation complexity, while CSI is obtained for each MU. By exploiting the temporal correlation, we propose an effective algorithm to dynamically select the MUs to be trained in each block, and determine the optimal training length. Simulation results show that the proposed selective training method achieves noticeable performance improvement compared to existing methods. In addition, as the MU density increases, the proposed method provides higher performance gains.

Notations: $(\cdot)^T$: transpose, $(\cdot)^H$: conjugate transpose, $(\cdot)^{-1}$: inverse, $|\cdot|$: determinant, $\mathrm{tr}(\cdot)$: trace, $\left\|\cdot\right\|_F$: Frobenius norm, $\mathbb{E}[\cdot]$: expectation, $\circ$: Hadamard product, $\mathrm{diag}(\cdot)$: diagonal matrix, $\mathbb{C}$: complex number, $\mathbb{Z}^+$: positive integer, $\mathrm{Var}(\cdot)$: variance, $\mathrm{card}(\cdot)$: cardinality.

\section{System Model and Problem Formulation}\label{system}
 We consider the uplink transmission in a TDD massive MIMO system with an $N$-antenna BS and $K$ single-antenna MUs, as shown in Fig.\ref{MIMO}. Uplink channel estimation is considered, where the BS obtains CSI through either training-based estimation, i.e., to estimate the CSI based on the received pilots sent by the MUs, or prediction, i.e., to predict the CSI from the previous estimates.

\begin{figure}[htbp]
\centering\includegraphics[height=4.7cm]{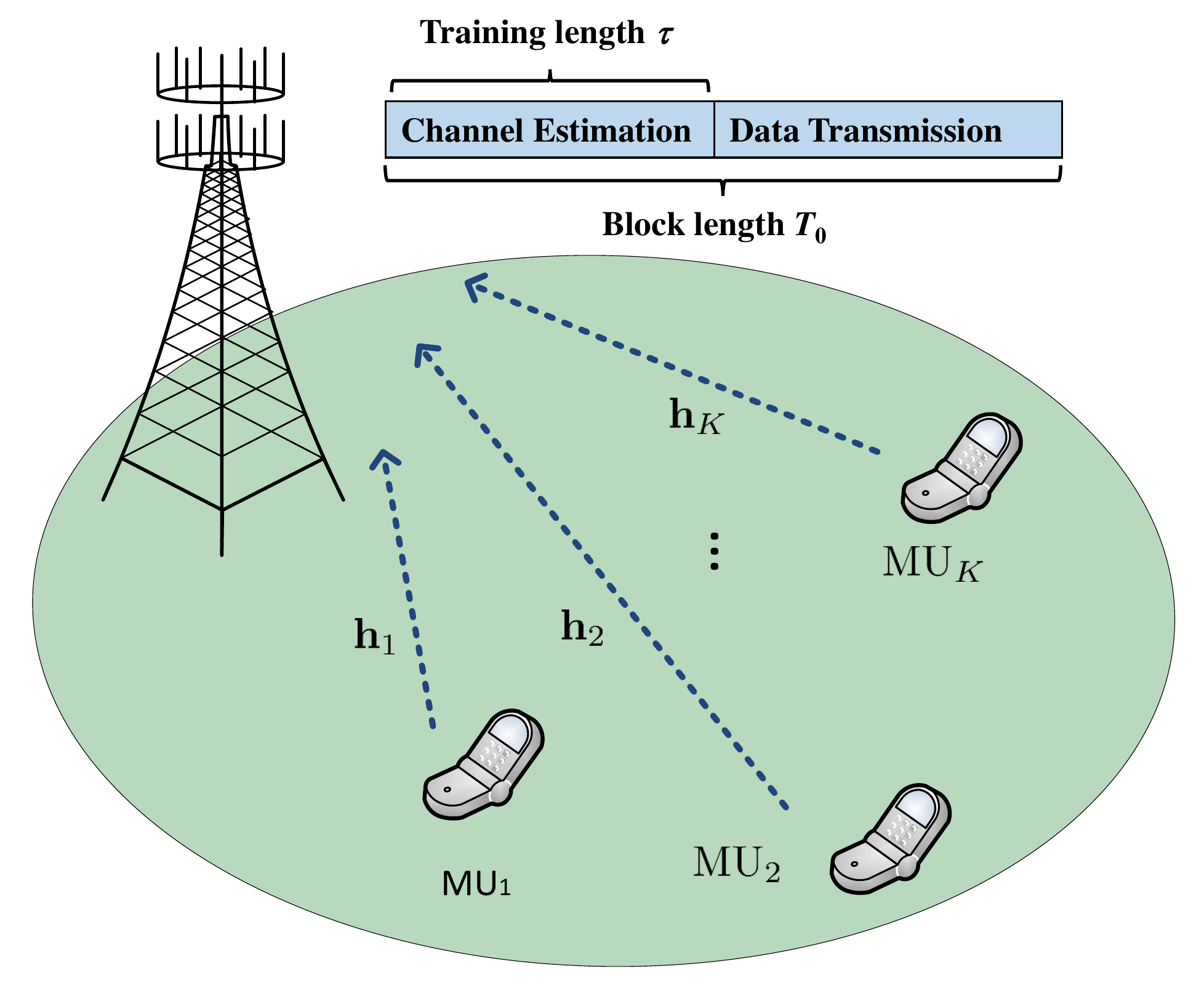}
\caption{A TDD massive MIMO system.}\label{MIMO}
\end{figure}

\subsection{Channel Model}
The channels are assumed to be block fading with coherence time $T_c$ and coherence bandwidth $B_c$, i.e., the channels remain static within each channel block, but vary among different channel blocks. Define ${T_0} \buildrel \Delta \over = {T_c}{B_c}$ as the \emph{block length}, which denotes the number of channel uses in each block. In particular, the channel vector from the $k$-th MU to the BS in the $b$-th channel block is denoted as $\mathbf{h}_{k,b} \in {\mathbb{C}^N}$. For convenience, we define the channel matrix for the $b$-th channel block as ${\mathbf{H}_b} \buildrel \Delta \over =[{\mathbf{h}_{1,b}},...,{\mathbf{h}_{K,b}}] \in {\mathbb{C}^{N \times K}}$. Motivated by the increasing density of mobile devices, we consider the scenarios in which the number of MUs is relatively large compared to the block length ${T_0}$, and denote $\alpha\buildrel \Delta \over = \frac{K}{T_0}$. It is worthwhile to note that such scenarios have not been addressed in existing studies \cite{massive_fundamentals}, although they are realistic and important to consider for future massive MIMO systems.

We focus on a low mobility environment and assume the channel spatial and temporal statistics remain unchanged within $J$ consecutive channel blocks. Specifically, the channel matrix $\mathbf{H}_b$ in the $b$-th channel block can be written as $\mathbf{H}_b = \mathbf{L}\circ\mathbf{G}_b,\forall b \in \mathcal{J}$,
where $\mathcal{J}\buildrel \Delta \over =\{0,1,...,J-1\}$, $\mathbf{G}_b=[g_{ik}]_b\in\mathbb{C}^{N \times K}$ and $\mathbf{L}=[l_{ik}]\in\mathbb{C}^{N\times K}$ represent the small-scale and large-scale fading channel coefficient matrix, respectively. Without loss of generality, we assume $\mathbf{H}_b$ evolves according to a first-order stationary Gauss-Markov process \cite{downlink_training}, i.e.,
\begin{equation}\label{GM}
\mathbf{H}_b=\mathbf{C}\circ\mathbf{H}_{b-1}+\mathbf{Z}_b,\forall b \in \mathcal{J},
\end{equation}
where $\mathbf{C}=[c_{ik}]\in\mathbb{C}^{N\times N}$ is the temporal correlation coefficient matrix, which depends on the channel instantiation interval and the maximum Doppler frequency according to Jake's model \cite{jakes}, and $\mathbf{Z}_b=[z_{ik}] \in \mathbb{C}^{N \times N}$ is an innovation process. For ease of notation, we denote $v_{ik}$ as the variance of $h_{ik}$, i.e., $v_{ik}=\mathrm{Var}(h_{ik})$, and that $z_{ik}\sim \mathcal{CN}(0, (1-c_{ik}^2)v_{ik})$ is independent from the channel realization history.

\subsection{Proposed Selective Training for Uplink Channel Estimation}\label{uplinktraining}
We define $\tau$, where $0<\tau<T_0$, as the \emph{training length}, i.e., the first $\tau$ channel uses will be utilized for training in each channel block. At the $b$-th channel block, the $k$-th MU sends the training sequence $\mathbf{x}_k\in \mathbb{C}^{1\times \tau}$. Thus, the received signal $\mathbf{Y}_{b}$ at the BS is given as
\begin{equation}\label{signalmodel}
\mathbf{Y}_b=\mathbf{H}_b \mathbf{X}_b+\mathbf{N}_b,
\end{equation}
where $\mathbf{Y}_b \in \mathbb{C}^{N \times \tau}$, $\mathbf{X}_b\buildrel \Delta \over =[\mathbf{x}_1^T,...,\mathbf{x}_K^T]^T\in\mathbb{C}^{K\times\tau}$ is the training matrix, and $\mathbf{N}_b=[n_{ik}] \in \mathbb{C}^{N \times \tau}$ denotes the additive Gaussian noise with unit variance.
To facilitate the analysis and practical implementation, we adopt orthogonal sequences for channel training:
\begin{equation}\label{othogonal}
\mathbf{X}_b \in \mathcal{X}\buildrel \Delta \over =\{\mathbf{F}:\mathbf{F}\in\mathbb{C}^{K\times\tau},\mathbf{F}\mathbf{F}^H=\tau\mathbf{I}_K\}.
\end{equation}

Channel training design for massive MIMO systems is a highly non-trivial task due to the huge amount of CSI to be obtained. If the BS were to perform training-based estimation for all the MUs, as in conventional methods, the training overhead with orthogonal training would become extremely heavy and would occupy most of the available radio resources. In order to reduce the training overhead, we propose a selective training method by leveraging the benefits of orthogonal pilots and exploiting the channel's temporal correlation. Specifically, in each channel block, the BS trains part of the MUs, and predicts the channel states for the remaining MUs according to the temporal correlation.
\subsubsection{Training-based Estimation}
In the training-based estimation, the BS estimates the CSI based on the received pilot symbols by scalar minimum mean-square error (MMSE) estimation \cite{optimal_channel}. In other words, the BS first obtains the received signals for the $k$-th MU:
\begin{equation}\label{orthogonal}
\begin{split}
\mathbf{s}_{k,b}=\mathbf{Y}_b \mathbf{x}_k^H=\tau \mathbf{h}_{k,b}+\mathbf{N}_b \mathbf{x}_k^H,k\in \mathcal{K}_T,
\end{split}
\end{equation}
where $\mathbf{N}_b \mathbf{x}_k^H \sim \mathcal{CN}(0,\tau\mathbf{I}_N)$, $\mathcal{K}$ and $\mathcal{K}_T$ denote the set of MUs and the set of MUs that are selected to be trained, respectively. After normalization, we have
$\mathbf{r}_{k,b} = \frac{1}{{\sqrt \tau  }}{\mathbf{s}_{k,b}}= \sqrt \tau  {\mathbf{h}_{k,b}} + {\mathbf{N}_{k,b}}$.
Thus, the scalar estimation channel from the $k$-th MU to the $i$-th receive antenna is given by ${r_{ik,b}} = \sqrt \tau  {h_{ik,b}} + {n_{ik,b}}$, where $n_{ik,b}\sim \mathcal{CN}(0,1)$. By decomposing $h_{ik,b}$ into the estimate and the estimation error, i.e., $h_{ik,b}=\widehat h^t_{ik,b}+\widetilde h^t_{ik,b}$, where $\widehat h^t_{ik,b}$ and $\widetilde h^t_{ik,b}$ are independent, we can compute the MMSE estimate of $h_{ik,b}$ given the observation $r_{ik,b}$:
\begin{equation}\label{scalarestimation}
\widehat h^t_{ik,b}= \tau v_{ik}(1+\tau v_{ik})^{-1} r_{ik,b}.
\end{equation}

\subsubsection{Linear Prediction}
For the MUs that do not send pilots during the current block, the BS will use the obtained CSI from the last channel block, $\widehat h_{ik,b-1}$, as prior information for the current channel block. Similar to training-based estimation, we decompose the $h_{ik,b}$ into the prediction and the prediction error, i.e., $h_{ik,b}=\widehat h^p_{ik,b}+\widetilde h^p_{ik,b}$, where $\widehat h^p_{ik,b}$ and $\widetilde h^p_{ik,b}$ are independent.
Based on $\widehat h_{ik,b-1}$ and the channel evolution equation in \eqref{GM}, a linear predictor is adopted to predict the CSI of the remaining MUs, i.e.,
\begin{equation}\label{predict}
\widehat h^p_{ik,b}=c_{ik}\widehat h_{ik,b-1},k\in\mathcal{K}\slash \mathcal{K}_{T},
\end{equation}
which is the best linear unbiased prediction (BLUP) for the first-order Gauss-Markov model \cite{predictor}.

Linear prediction requires no training overhead and lower computational complexity, while training-based estimation provides more accurate CSI according to \eqref{scalarestimation}, \eqref{predict}. Thus, we balance the use of the training-based estimation and linear prediction in each channel block in the proposed selective training method.
\subsection{Data Transmission}\label{transmission}
In the data transmission phase, all the MUs send their data simultaneously, and the received signal $\mathbf{Y}_d$ at the BS is given as
\begin{equation}\label{data}
{\mathbf{Y}_d} = \widehat{\mathbf{ H}}(\tau ){\mathbf{X}_d} + \underbrace {\widetilde{\mathbf{ H}}(\tau ){\mathbf{X}_d} + \mathbf{N}}_{\mathbf{Q}(\tau )}.
\end{equation}
The BS jointly decodes the data signals from all MUs based on the obtained channel $\mathbf{\widehat H}(\tau)$, i.e., the BS treats $\mathbf{\widehat H}(\tau)\mathbf{X}_d$ as the desired signal and $\mathbf{Q}(\tau)$ as the equivalent noise, which includes the additive Gaussian noise, channel estimation error and channel prediction error.
\subsection{Capacity Maximization Problem}\label{2D}
According to Section \ref{uplinktraining} and Section \ref{transmission}, the channel capacity in the data transmission phase depends on the training length $\tau$ and the trained MU set $\mathcal{K}_T$, which can be expressed as $C(\tau,\mathcal{K}_T)=\sup\nolimits_{P_{X_d}(\cdot)} \text{\ }\frac{1}{K}I(\mathbf{Y}_d, \mathbf{X}_d; \mathbf{\widehat H}(\tau,\mathcal{K}_T))$. Considering that only a fraction of the total coherence block length, i.e., $(1-\frac{\tau}{T_{0}})$, is used for data transmission, the effective capacity is given by $(1-\frac{\tau}{T_0})C(\tau,\mathcal{K}_T)$. To enable practical implementation, we assume the training length $\tau$ is the same in all channel blocks, while $\mathcal{K}_T$ is designed dynamically for each block. Thus, we adopt the average effective capacity over $J$ channel blocks as the objective function. For simplicity, we assume the MUs' locations and temporal correlation coefficients are static within the considered $J$ blocks as prior information. As a result, the capacity maximization problem can be formulated as:\\
\begin{equation}\label{originalproblem}
\begin{array}{l}
\text{\ }\mathop {\max}\limits_{\tau } \text{\ } \frac{1}{J}\sum \limits_{b=1}^J (1-\frac{\tau}{T_0})C_b(\tau,\mathcal{K}_T^{b,\tau}) )\\
\text{\ \ \ }\mathrm{s.t.} \text{\ \;} 0\le\tau\le T_0,
\end{array}
\end{equation}
where $\mathcal{K}_T^{b,\tau}$ is the optimal training set for a given $(b,\tau)$. Designing $\mathcal{K}_T^{b,\tau}$ and optimizing $\tau$ are two important components of the proposed selective training method, which will be elaborated in the next section.
\section{Selective Training with Dynamic User Selection}\label{part3}
In this section, we will first introduce a lower bound of the channel capacity considering estimation error and prediction error, which will be used as the performance metric for the later training design. We will then propose a dynamic user selection (DUS) method to determine $\mathcal{K}_T^{b,\tau}$ for given $b$ and $\tau$, and then optimize the training length $\tau$.

\subsection{A Lower Bound of the Channel Capacity}
 In the $b$-th block, as the capacity of the channel described by \eqref{data} remains unknown \cite{optimal_channel}, we will use a lower bound of the normalized capacity $(1-\frac{\tau}{T_0})C_b(\tau,\mathcal{K}_T^{b,\tau})$ to evaluate the system performance, which will be called the {\emph {achievable rate}}, denoted as $R_b(\tau,\mathcal{K}_T^{b,\tau})$. It is obtained by regarding the term $\mathbf{\widehat H}_b(\tau,\mathcal{K}_T^{b,\tau})$ as the actual channel matrix and the equivalent noise $\mathbf{Q}_b(\tau,\mathcal{K}_T^{b,\tau})$ as independent complex Gaussian noise with covariance matrix $\mathbf{K}_{Q,b}(\tau,\mathcal{K}_T^{b,\tau})\in \mathbb{R}_{+}^{N \times N}$ \cite{how_much}, i.e.,
\begin{equation}\label{totalnoise}
\begin{split}
{\mathbf{K}_{Q,b}}(\tau,\mathcal{K}_T^{b,\tau})
&= \mathbb{E}\left[ {\mathbf{Q}_b{\mathbf{Q}_b^H}} \right]\\
&= \mathrm{diag}\left( \bigg\{{1 + \sum\limits_{k = 1}^K {{{\widetilde v}_{ik,b}}(\tau,\mathcal{K}_T^{b,\tau} )} } \bigg\}_{i = 1}^N\right).
\end{split}
\end{equation}
We define
$\mathbf{\overline H}_b  \buildrel \Delta \over = \mathbf{K}_{Q,b}^{ - \frac{1}{2}}(\tau,\mathcal{K}_T^{b,\tau} )\mathbf{\widehat H}_b(\tau,\mathcal{K}_T^{b,\tau} )$.
Considering the fraction for data transmission $(1-\frac{\tau}{T_{0}})$, the achievable rate per user can be written as
\begin{small}
\begin{equation}\label{rate}
\begin{split}
R_b(\tau,\mathcal{K}_T^{b,\tau} ) =
  \left( {1 - \frac{\tau }{{{T_0}}}} \right)\frac{1}{K}
 {\mathbb{E}_{\widehat {\mathbf{H}}_b}}\left[ {\log \left| {{\mathbf{I}_N} + \mathbf{\overline H}_b \mathbf{\overline H}_b {^H}} \right|} \right].
\end{split}
\end{equation}
\end{small}
Thus, the capacity maximization problem is reformulated as
\begin{equation}\label{secondproblem1}
\begin{array}{l}
\text{\ }\mathop {\max}\limits_{\tau} \text{\ } \frac{1}{J}\sum\limits_{b=1}^J R_b(\tau,\mathcal{K}_T^{b,\tau} )\\
\text{\ \ \ }\mathrm{s.t.} \text{\ \ } 0\le \tau \le T_0.
\end{array}
\end{equation}
However, it is still challenging to solve due to the expectation involved in $R_b(\tau,\mathcal{K}_T^{b,\tau})$, and the combinatorial structure of $\mathcal{K}_T^{b,\tau}$. In the following, for a given $(b,\tau)$, we will first propose a dynamic user selection method to determine $\mathcal{K}_T^{b,\tau}$, and provide an accurate approximation for $R_b(\tau,\mathcal{K}_T^{b,\tau})$ to search for the optimal training length $\tau$.

\subsection{Dynamic User Selection}
 In the $b$-th block, we maximize $R_b(\tau,\mathcal{K}_T^{b,\tau})$, which consists of two parts:  $(1 - \frac{\tau }{{{T_0}}})\frac{1}{K}$ and ${\mathbb{E}_{\widehat {\mathbf{H}}_b}}\left[ {\log \left| {{\mathbf{I}_N} + \mathbf{\overline H}_b \mathbf{\overline H}_b {{}^H}} \right|} \right]$, where the first part is only related to $\tau$, while the second part is related to $\tau$ and $\mathcal{K}_T^{b,\tau}$.
In this subsection, we consider a fixed training length $\tau$, and obtain $\mathcal{K}_T^{b,\tau}$ by maximizing the second part via user selection:
\begin{equation}\label{secondproblem}
\begin{array}{l}
\text{\;}\mathop {\max}\limits_{ \mathcal{K}_T^b} \text{\ } {\mathbb{E}_{\widehat H_b}}\left[ {\log \left| {{\mathbf{I}_N} + \mathbf{\overline H}_b (\tau,\mathcal{K}_T^b) \mathbf{\overline H}_b {{(\tau,\mathcal{K}_T^b )}^H}} \right|} \right]\\
\ \ \ \mathrm{s.t.} \text{\ \ } \mathcal{K}_T^b\in \mathcal{K}\\
\ \ \ \ \ \ \ \ \; \mathrm{card}(\mathcal{K}_T^b) \le \tau.\\
\end{array}
\end{equation}
Similar to \eqref{secondproblem1}, \eqref{secondproblem} is still intractable due to the complex objective function. Instead, we will minimize the term inside the logarithmic function, i.e., $\left| {{\mathbf{I}_N} + \mathbf{\overline H}_b (\tau,\mathcal{K}_T^{b}) \mathbf{\overline H}_b {{(\tau,\mathcal{K}_T ^{b})}^H}} \right|$. By substituting $ \mathbf{K}_{Q,b}^{ - \frac{1}{2}}\mathbf{\widehat H}_b$ into $\mathbf{\overline H}_b$, we obtain $\left| {{\mathbf{I}_N} + \mathbf{K}_{Q,b}^{ - \frac{1}{2}}\mathbf{\widehat H}_b\mathbf{\widehat H}_b^H \mathbf{K}_{Q,b}^{ - \frac{1}{2},H}} \right|$, where the equivalent noise $\mathbf{K}_{Q,b}$ has a significant influence. Recall that, the antennas at the BS are co-located in massive MIMO systems. Thus, the entries of the channel vector from the $k$-th MU to the BS have an identical variance, i.e., $\widetilde v_{ik,b}=\beta_{k,b},\forall i$, where $\beta_{k,b},\forall k$ is a constant related to $(b,\tau,\mathcal{K}_T^b)$ and it changes in different channel blocks. We define $\beta_b$ as
\begin{equation}\label{beta}
\beta_b (\tau,\mathcal{K}_T^b)
\buildrel \Delta \over = \sum\limits_{k = 1}^K {{\beta _{k,b} (\tau,\mathcal{K}_T^b)}}
={\sum\limits_{k \in {\mathcal{K}_T^b}} {{\beta_{k,b}^t(\tau)} + \sum\limits_{k \in \mathcal{K}\backslash {\mathcal{K}_T^b}} { \beta_{k,b}^p} } },\nonumber
\end{equation}
where $\beta_{k,b}^t(\tau)$ denotes the estimated channel variance of the trained MUs and $\beta_{k,b}^p$ denotes the predicted channel variance of the remaining MUs. Then, the covariance matrix of the equivalent noise can be written as ${\mathbf{K}_{Q,b}}(\tau,\mathcal{K}_T^b)= (1+\beta_b(\tau,\mathcal{K}_T^b))\mathbf{I}_N$. As a result, the term $\left| {{\mathbf{I}_N} + \mathbf{K}_{Q,b}^{ - \frac{1}{2}}\mathbf{\widehat H}_b\mathbf{\widehat H}_b^H \mathbf{K}_{Q,b}^{ - \frac{1}{2},H}} \right|$ can be rewritten as
\begin{equation}
  \left| {{\mathbf{I}_N} +(1+ \beta_b(\tau,\mathcal{K}_T^b))^{-1}\mathbf{\widehat H}_b\mathbf{\widehat H}_b^H} \right|.
  \end{equation}
 Denote the eigenvalues of $\mathbf{\widehat H}_b\mathbf{\widehat H}_b^H$ as $\{\lambda_1,...,\lambda_K\}$, where $\lambda_k\ge 0, \forall k\in \mathcal{K}$. Thus, the determinant can be written as $\prod\limits_{k = 1}^K {\left( {1 + \frac{1}{{1 + {\beta _b}(\tau ,{\mathcal{K}_T^b})}}{\lambda _k}} \right)}$. For tractability, we assume $\{\lambda_1,...,\lambda_K\}$ does not depend on $\mathcal{K}_T^b$, which is reasonable when $\beta^t_{k,b}(\tau)$ is close to $\beta^p_{k,b}$, i.e., the accuracy of the predicted CSI is close to that of the estimated CSI. As will be shown through simulations, this will help to develop a very effective user selection method. Thus, to improve the achievable rate, we consider minimizing $\beta_b(\tau,\mathcal{K}_T^b)$. In the $b$-th block, for a given $\tau$, we shall solve the following problem:
 \begin{equation}\label{optimization}
\begin{array}{l}
\mathop {\min}\limits_{ \mathcal{K}_T^b} \text{\ } { \sum\limits_{k \in {\mathcal{K}_T^b}} {{\beta_{k,b}^t(\tau)} + \sum\limits_{k \in \mathcal{K}\backslash {\mathcal{K}_T^b}} { \beta_{k,b}^p} } } \\
\  \mathrm{s.t.} \text{\ \ } \mathcal{K}_T^b\in \mathcal{K}\\
\text{\ \ \ \ \ \ \;} \mathrm{card}(\mathcal{K}_T^b) \le \tau.\\
\end{array}
\end{equation}
 We will first specify the terms in the objective function. According to \eqref{scalarestimation} and \eqref{predict}, the variances of the obtained channel based on training and prediction are respectively given as
${\widehat v^t_{ik,b}}(\tau) = \mathrm{Var}({\widehat h^t_{ik,b}}) = \frac{{\tau v_{ik}^2}}{{\tau {v_{ik}} + 1}}$ and
${\widehat v^p_{ik,b}} = \mathrm{Var}({\widehat h^p_{ik,b}}) = c^2_{ik}{\widehat v_{ik,b-1}}$.
According to the previous assumption, the channel statistics remain unchanged, i.e., $v_{ik}$ is constant in the $J$ channel blocks. Due to the orthogonality principle for the MMSE estimates and the independent innovation process, we have the relationship: $v_{ik}=\mathrm{Var}(\widehat {h}_{ik,b})+\mathrm{Var}(\widetilde {h}_{ik,b})=\widehat {v}_{ik,b}+\widetilde {v}_{ik,b}$. Thus, the variances of the estimation error and the prediction error are given by $\beta_{k,b}^t(\tau)={\widetilde v^t_{ik,b}}(\tau ) = v_{ik}-{\widehat v^t_{ik,b}}(\tau )$ and $\beta_{,b}^p={\widetilde v^p_{ik,b}} = v_{ik}-\widehat{v}_{ik,b}^p$, respectively.

In \eqref{optimization}, the set of the trained MUs in the $b$-th channel block needs to be decided. For this purpose, we evaluate the difference between the estimation error and the prediction error for each MU, which is defined as $
\Delta \beta _{k,b}(\tau) \buildrel \Delta \over =\beta_{k,b}^p-\beta_{k,b}^t(\tau),k\in\mathcal{K}$,
where a large $\Delta \beta_{k,b}(\tau)$ means that training for the $k$-th MU can more effectively reduce the error of the obtained CSI. With this setup, we can obtain the solution of problem \eqref{optimization} by Algorithm \ref{selection}. The following result verifies that Algorithm 1 will give the optimal solution to problem \eqref{optimization}.
\begin{algorithm}[h]
\caption{Dynamic User Selection for Problem \eqref{optimization}}
\label{selection}
\begin{algorithmic}[1]
\REQUIRE
$\tau,K,v_{ik}, \widehat {v}_{ik,b-1}$ and $\widetilde {v}_{ik,b-1}$, where $i=1,...,N;k=1,...,K$
\STATE Let $\mathcal{K}=\{ 1,...,K\}$;
\IF{$K\le \tau$}
\STATE $\mathcal{K}_T^{b,\tau}=\mathcal{K}$;
\ELSE
\STATE Calculate the estimation error $\beta_{k,b}^t(\tau),k=1,...,K$;
\STATE Calculate the prediction error $\beta_{k,b}^p,k=1,...,K$;
\STATE Calculate $\Delta \beta _{k,b}(\tau) \buildrel \Delta \over =\beta_{k,b}^p-\beta_{k,b}^t(\tau), k=1,...,K$;
\STATE Obtain $\mathcal{K}_T^{b,\tau}$ by selecting the first $\tau$ largest $\Delta \beta _{k,b}(\tau)$;
\ENDIF
\RETURN $\mathcal{K}_T^{b,\tau}$ and $\mathcal{K}\backslash \mathcal{K}_T^{b,\tau}$
\end{algorithmic}
\end{algorithm}

\begin{lma}
For a fixed $\tau$, optimal solution for problem \eqref{optimization} can be given by Algorithm \ref{selection}.
\end{lma}

\begin{proof}
Since the values of $\tau, K, v_{ik}, \widehat{v}_{ik,b-1}$ and $\widetilde{v}_{ik,b-1}$ are known, in the $b$-th channel block, we first assume that all the CSI is predicted, i.e., $\mathcal{K}_T^b=\varnothing$. The initial value of the objective function of problem \eqref{optimization} is: $F_0 \buildrel \Delta \over ={1 + \sum\nolimits_{k \in \mathcal{K}} { \beta_{k,b}^p}}=1+\sum\nolimits_{k \in \mathcal{K}}[{\widetilde v_{ik,b-1}}(\tau )+ (1-c^2_{ik}){\widehat v_{ik,b-1}}(\tau )]$,
which is a constant. After selecting part of the MUs for training, the value of the objective function becomes
\begin{small}
\[
\begin{split}
F_k&={1 + \sum\limits_{k \in {\mathcal{K}_T^b}} {{\beta_{k,b}^t}(\tau ) + \sum\limits_{k \in \mathcal{K}\backslash {\mathcal{K}_T^b}} { \beta_{k,b}^p} } }\\
&={1 + \sum\limits_{k \in {\mathcal{K}
_T^b}} {({\beta_{k,b}^p-\Delta \beta_{k,b}(\tau)}) + \sum\limits_{k \in \mathcal{K}\backslash {\mathcal{K}_T^b}} { \beta_{k,b}^p} } }\\
&={1 + \sum\limits_{k \in \mathcal{K}} {\beta_{k,b}^p}-\sum\limits_{k\in\mathcal{K}
_T^b}{\Delta \beta_{k,b}(\tau)}}=F_0-\sum\limits_{k\in\mathcal{K}_T^b}{\Delta \beta_{k,b}(\tau)}\\
\end{split}
\]
\end{small}
Thus minimizing $F_k$ is equivalent to maximizing $\sum\nolimits_{k\in\mathcal{K}
_T^b}{\Delta \beta_{k,b}(\tau)}$. Since $\mathcal{K}_T^b\in\mathcal{K}$, $\mathrm{card}(\mathcal{K}_T^b)\le\tau$, we have $\mathrm{card}(\mathcal{K}_T^b)\le \min(\tau,K)$. Therefore, $\sum\nolimits_{k\in\mathcal{K}_T^b}{\Delta \beta_{k,b}(\tau)}$ is maximized by selecting the first $\min(\tau,K)$ largest $\Delta \beta_{k,b}(\tau)$, i.e., problem \eqref{optimization} is solved by selecting the first $\min(\tau,K)$ largest $\Delta \beta_{k,b}(\tau)$ to form the subset $\mathcal{K}_T^{b,\tau}$.
\end{proof}
\subsection{Training Length Optimization}
So far, we have obtained $\mathcal{K}_T^{b,\tau}$ for a given $(b,\tau)$. Nevertheless, optimizing $\tau$ is still intractable, as it is difficult to compute the value of the achievable rate $R_b(\tau,\mathcal{K}_T^{b,\tau})$. As a result, we resort to an approximation based on the theory of large random matrices \cite{optimal_channel}, \cite{equivalent}, where the block index $b$ and $\mathcal{K}_T^{b,\tau}$ will be omitted temporally.
\begin{lma}[Deterministic Equivalent]\label{lemma}
Define $\ {\overline v _{ik}}(\tau ) = {{{{\widehat v}_{ik}}(\tau )}}({{1 + \sum\nolimits_{\ell  = 1}^K {{{\widetilde v}_{ik}}(\tau )} }})^{-1}$ and consider the following $N\times N$ matrices
${\mathbf{D}_k}(\tau ) = \mathrm{diag}({\overline v _{1k}}(\tau ),...,{\overline v _{Nk}}(\tau )),\;k = 1,...,K$. Let $\tau>0$, and assume that $K$ and $N$ satisfy $0 < \mathop {\lim \inf }\limits_{K \to \infty } \frac{N}{K} \le \mathop {\lim \sup }\limits_{K \to \infty } \frac{N}{K} < \infty$ and $0 < {\overline v _{ik}}(\tau ) < {v_{\max }} < \infty ,\;\forall i,k$. The equivalent approximation of the achievable rate \eqref{rate} is given as:
\begin{equation}\label{closedform}
\begin{split}
\begin{array}{l}
\overline R(\tau ) = \left( {1 - \frac{\tau }{{{T_0}}}} \right)\frac{1}{K}\left[ {\sum\limits_{k = 1}^K {\log \left( {1 + \frac{1}{K} \mathrm{tr} {{\bf{D}}_k}(\tau ){{\bf{T}}_P}} \right)} } \right.\\
\text{\;\;\;\;\;\;\;\;\;\;\;\;\;}\left. {{\rm{   }}- \log \det \left( {\frac{1}{K}{{\bf{T}}_P}} \right) - \sum\limits_{k = 1}^K {\frac{{\frac{1}{K} \mathrm{tr}{{\bf{D}}_k}(\tau ){{\bf{T}}_P}}}{{1 + \frac{1}{K} \mathrm{tr} {{\bf{D}}_k}(\tau ){{\bf{T}}_P}}}} } \right],
\end{array}
\end{split}
\end{equation}
where $\mathbf{T}_P$ is given by an implicit equation:
\begin{equation}\label{TQ}
\mathbf{T}_P = {\left( {\frac{1}{K}\sum\limits_{k = 1}^K {\frac{{{\mathbf{D}_k}(\tau )}}{{1 + \frac{1}{K} \mathrm{tr} {\mathbf{D}_k}(\tau )\mathbf{T}_P}} + } \frac{1}{K}{\mathbf{I}_N}} \right)^{ - 1}},
\end{equation}
which admits a unique solution, $\mathbf{T}_P=\mathrm{diag}(t_1,...,t_N)$. It has been shown in \cite{equivalent} that the solution of \eqref{TQ} is unique and can be found by using an iterative algorithm.
Then, the following results hold:
$\mathop {\lim }\nolimits_{K \to \infty } \left[ {R(\tau ) - \overline R (\tau )} \right] = 0$ and $\mathop {\lim }\nolimits_{K \to \infty } \left[ {\tau^* - \overline {\tau}^* } \right] = 0$, where $
{\tau ^*} = \mathop {\max }\nolimits_{\tau  \in [0,{T_0}]} R(\tau )$ and $\overline {\tau} ^*=\mathop{\max}\nolimits_{\tau \in [0,T_0]} \overline {R}(\tau)$.
\end{lma}

Thus, the capacity maximization problem \eqref{originalproblem} becomes
\begin{equation}\label{reproblem}
\begin{array}{l}
\text{\ }\mathop {\max}\limits_{\tau} \text{\ } \frac{1}{J}\sum\limits_{b=1}^J \overline R_b(\tau,\mathcal{K}_T^{b,\tau})\\
 \text{\ \ \ }\mathrm{s.t.} \text{\ \ } 0 \le \tau \le T_0,
\end{array}
\end{equation}
where $\mathcal{K}_T^{b,\tau}$ is determined by Algorithm \ref{selection}, which makes the search for the optimal training length efficient and practical without the need for searching all the combinations of the subset of $\mathcal{K}$ and training length $\tau$.

So far, we have considered the scenarios where the MU locations are fixed and known, i.e., the channel statistics are fixed. In practice, for different MU locations, the optimal training length may be different. To make the proposed method practical, we can find a common training length for a given value of $K$, but for different MU locations. This common optimal training length can be obtained by maximizing the average achievable rate over different MU locations, where such averaging is highly intractable and can be obtained via simulations. Such an approach is feasible, as the searching for the optimal $\tau$ can be done offline, and it will be effective with the help of Algorithm \ref{selection}. Once the optimal training length is determined, for each channel block, $\mathcal{K}_T^{b,\tau}$ can be dynamically designed based on the optimal training length and channel statistics.

\vspace{-10pt}

\section{Simulation Results}\label{simulation}
In this section, we simulate the proposed selective training scheme for a massive MIMO system. We consider the first-order Gauss-Markov channel model mentioned in Section \ref{system}, where
$g_{ik,b}\sim \mathcal{CN}(0,1)$ denotes the small-scale fading coefficient, and $J=11$ denotes the number of channel blocks for each realization of the MU locations. For the large-scale fading coefficient, $l_{ik}=(\frac{d_k}{d_0})^{-2}$, i.e., the path-loss exponent is 4, where $d_0=1\mathrm{km}$ is the reference distance and $d_k$ denotes the distance from the $k$-th MU to the BS. The innovation process in \eqref{GM} is given by ${z_{ik,b}} = \sqrt {1 - {c^2}} {\left( {\frac{{{d_k}}}{{{d_0}}}} \right)^{ - 2}}{u_{ik,b}}$ \cite{downlink_training}, where ${u_{ik,b}} \sim \mathcal{CN}(0,1)$ and $c_{ik}=c,\forall{i,k}$, with $c={J_0}(2\pi {f_D}\kappa)$ as the temporal correlation coefficient based on the Jake's model, where $J_0$ is the 0-th order Bessel function of the first kind, $f_D$ denotes the maximum Doppler frequency and $\kappa$ represents the channel instantiation interval. We set $c=0.9881$, as in \cite{downlink_training}. We assume the small-scale fading coefficients, innovation process and the observation noise within each block are mutually independent. For selective training methods, we assume that the BS performs full training in the first block and performs selective training from the $2$nd to the $J$-th block. In the training phase, we use the orthogonal training sequence \eqref{othogonal}. In the data transmission phase, we set the received $\mathrm{SNR_0}$ of the MUs located at $d_0$ to 0dB, i.e., $\mathrm{SNR}_0=0dB$, and thus the received SNR of the $k$-th MU can be calculated: $\mathrm{SNR}_k=\mathrm{SNR}_0-40{\log _{10}}(\frac{{{d_k}}}{{{d_0}}}),\forall k\in\mathcal{K}$. The performance is measured by the lower bound of the achievable rate \eqref{rate}.
\subsection{Benchmarks}
For comparison, we first introduce the conventional full training method and a user-scheduling method, as well as a selective training method, but with random user selection. Our proposed dynamic user selection method is denoted as DUS.
\subsubsection{Full Training (FT)} The BS performs training for all the MUs and serves all the MUs in each block, which is commonly assumed in previous works \cite{optimal_channel}.
\subsubsection{Random User Selection (RUS)} This is a selective training scheme. In each block, the BS performs training for part of the MUs, which are selected uniformly and randomly from all the MUs, and the remaining MUs' CSI is predicted according to \eqref{predict}. With the obtained CSI, the BS serves all the MUs simultaneously.
\subsubsection{User Scheduling (US)} One method to reduce training overhead is to train and serve a subset of the MUs, which is similar to multiuser scheduling in conventional multiuser MIMO channels. But differently, the scheduling of users should be based on channel statistics, and here we propose to train and serve the MUs that are closest to the BS, which will give high spectral efficiency.
\vspace{-10pt}
\begin{figure}[htbp]
\centering\includegraphics[height=6.3cm]{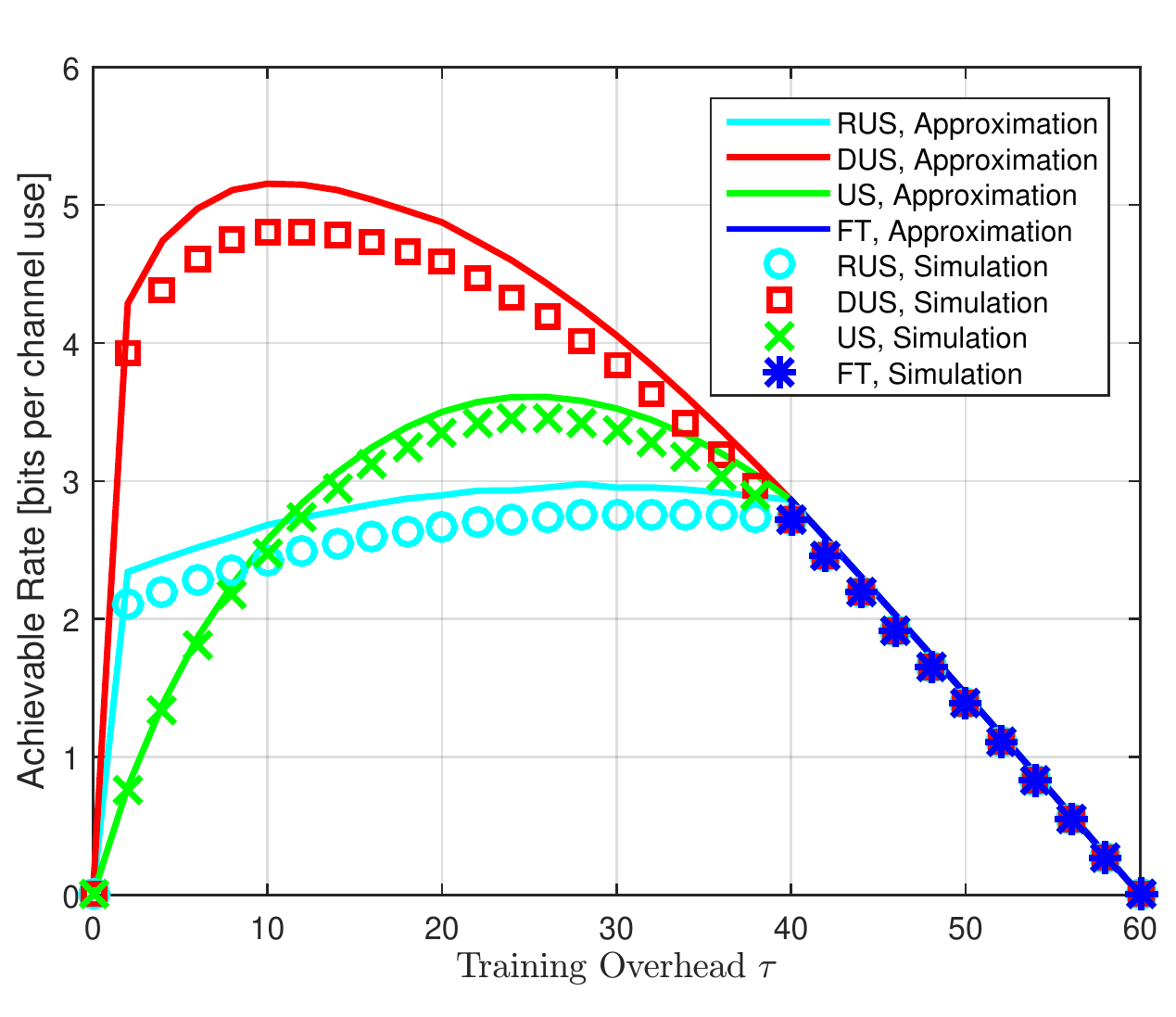}
\vspace{-10pt}
\caption{Average achievable rate vs. Training overhead.}\label{trainingoverhead}
\end{figure}
\vspace{-10pt}
\subsection{Average Achievable Rate vs. Training Overhead}
Consider a massive MIMO system with one 100-antenna BS located at the center $(0,0)$ and $K=40$ single-antenna MUs whose positions are uniformly and independently generated in a circular region with $\mathrm{radius}=1\mathrm{km}$. The block length is $T_0=60$, i.e., $\alpha=\frac{K}{T_0}=\frac{2}{3}$. Fig.\ref{trainingoverhead} shows two sets of curves: one is based on the approximations in \eqref{closedform}, and the other is the simulation results averaged over $10^4$ randomly generated channel realizations. For the RUS and US, the BS selects $\mathrm{min}(\tau,K)$ MUs for a given $\tau$, and for RUS, there are $10^3$ random user selection realizations in each block. As a result of the orthogonal training assumption, the curves for different methods overlap when $40\le \tau \le 60$, and the curves for the FT case start from $\tau=40$. From this figure, we can observe that $\overline {R}(\tau)$ is a good approximation of ${R}(\tau)$. We can also see that the average achievable rate decreases almost linearly for the conventional FT methods as the training overhead increases, which shows that the training overhead incurs a significant throughput degradation in such a setting. By optimizing the training overhead, we can get a better performance. The average throughput corresponding to the optimal training length of the proposed DUS training method is the largest among all the methods. The gap between selective training with DUS and the FT case results from the different training overheads, as that of the former is greatly reduced. Besides, due to only training for part of MUs, the channel estimation complexity is much reduced. On the other hand, the gap between the DUS and RUS methods shows that the proposed DUS method achieves more effective training by dynamically selecting MUs via Algorithm \ref{selection}. Compared to the US, selective training with DUS achieves much higher spectrum efficiency by exploiting the temporal correlation to obtain the CSI for all the MUs and serving all the MUs in each block.
\begin{figure}[htbp]
\centering\includegraphics[height=6.3cm]{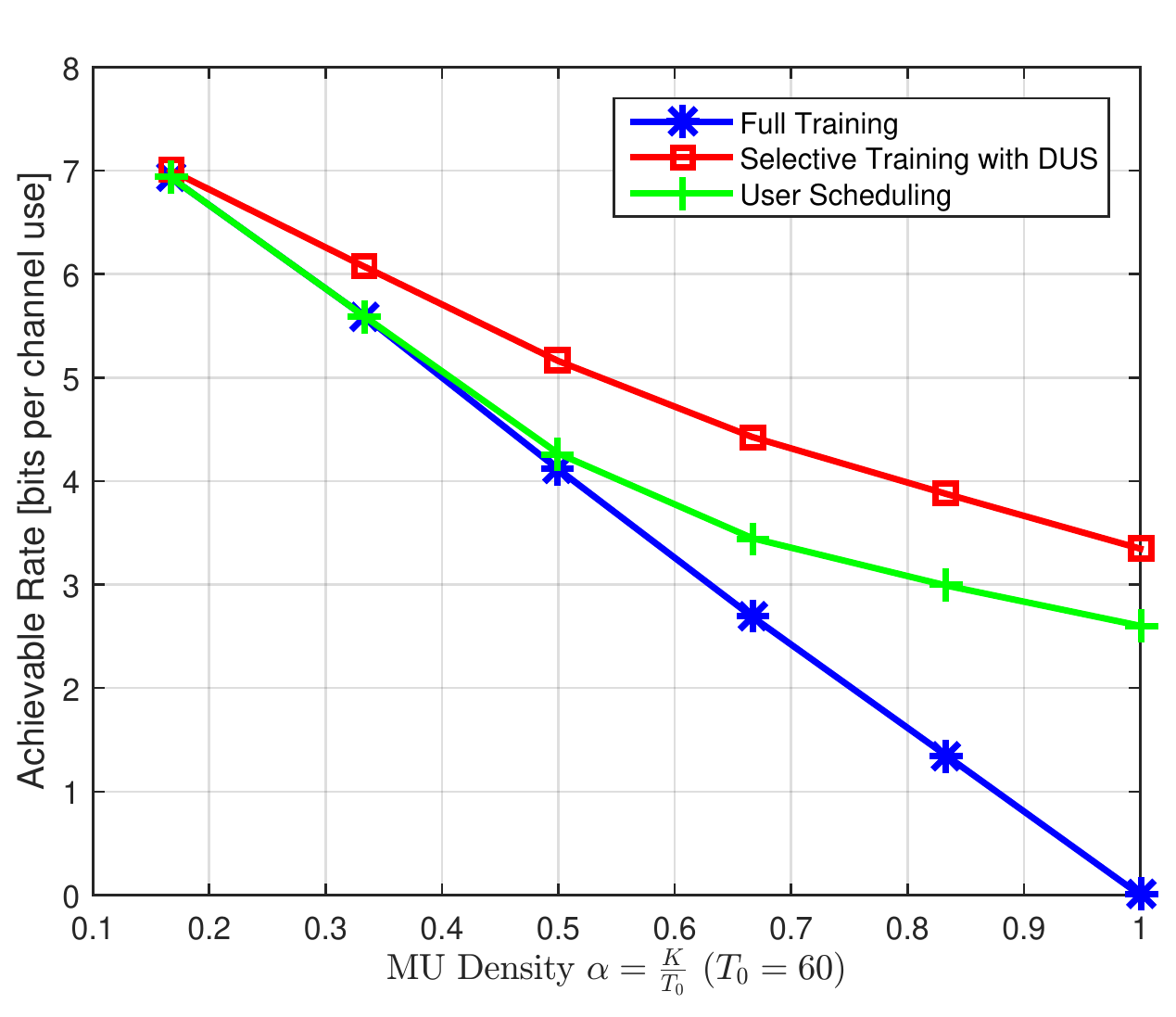}
\vspace{-10pt}
\caption{Average achievable rate vs. User density.}\label{CapacityDensity}
\end{figure}

\subsection{Average Achievable Rate vs. User Density}
Consider a massive MIMO system with one 100-antenna BS located at the center $(0,0)$ and $K$ single-antenna MUs that are uniformly and independently distributed in a circular region with $\mathrm{radius}=1\mathrm{km}$. Note that the achievable rate is averaged over $10^2$ randomly generated MUs' distribution realizations and $10^3$ randomly generated channel realizations. The block length is $T_0=60$. $K$ varies from 10 to 60, i.e., $\alpha$ varies from $\frac{1}{6}$ to $1$, to explore the influence of the MU density on the different methods. The RUS method is omitted due to its poor performance. We firstly search the optimal training length $\tau$ corresponding to the maximal approximation $\overline R(\tau)$ for different methods and parameters. Then the obtained $\tau$ is used for the simulation, based on which, the $\mathcal{K}_T$ is designed for the selective training with DUS in each block. From Fig.\ref{CapacityDensity}., we observe that the average achievable rate of the conventional FT case decreases almost linearly as $\alpha$ increases, and thus it is not applicable with dense MUs. The proposed selective training scheme with DUS performs the best among all the methods, which shows the effectiveness of our proposal. As $K$ increases, the performance gain of the proposed method becomes larger, i.e., the selective training is more effective, especially in the networks with dense MUs. It is worthwhile to mention that, the offline search for $\tau$ makes the proposed method practical, and selective training will also significantly reduce the estimation complexity.
\vspace{-5pt}
\section{Conclusions}\label{conclusion}
In this paper, we investigated the uplink training for massive MIMO systems with time-correlated channels. A selective training method with dynamic user selection was proposed, which can help to greatly reduce the training overhead by training only part of the MUs in each block. The proposed selective training method was shown to perform much better than conventional full training methods. Overall, this study has provided some promising results for massive MIMO systems with dense MUs, which previously was believed not to be workable due to the huge training overhead. The results of this paper have shown that with innovative training schemes, and by exploiting the temporal correlation of channels, it is possible to support MUs with the number comparable to the channel coherent length. Further investigation will be needed to continue this line of research, to make the proposed method more practical and extend it to other systems.

%
%



\bibliographystyle{IEEEtran}
%

\vspace{-10pt}
\begin{thebibliography}{1}


\bibitem{massive_mimo}
E. Larsson, F. Tufvesson, O. Edfors, and T. Marzetta, ``Massive MIMO for next generation wireless systems,'' \emph{IEEE Commun. Mag.}, vol. 52, no. 2, pp. 185-195, Feb. 2014.

\bibitem{inter_cell_interference}
F. Fernandes, A. Ashikhmin, and T. L. Marzetta, ``Inter-cell interference in noncooperative TDD large scale antenna systems,'' \emph{IEEE J. Sel. Areas Commun.}, vol. 31, no. 2, pp. 192-201, Feb. 2013.

\bibitem{massive_fundamentals}
E. Bj\"{o}rnson, E. G. Larsson, T. L. Marzetta, ``Massive MIMO: ten myths and one critical question,'' \emph{IEEE Commun. Mag.}, vol. 54, no. 2, pp. 114-123, Feb. 2016.

\bibitem{user_capacity}
J. Shen, J. Zhang, and K. B. Letaief, ``Downlink user capacity of massive MIMO under pilot contamination,'' \emph{IEEE Trans. Wireless Commun.}, vol. 14, no. 6, pp. 3183-3193, Jun. 2015.

\bibitem{downlink_training}
J. Choi, D. J. Love, and P. Bidigare, ``Downlink training techniques for FDD massive MIMO systems: Open-loop and closed-loop training with memory,'' \emph{IEEE J. Sel. Topics Signal Process.}, vol. 8, no. 5, pp. 802-814, Oct. 2014.


\bibitem{multiuser_MIMO}
G. Caire, N. Jindal, M. Kobayashi, and N. Ravindran, ``Multiuser MIMO achievable rates with downlink training and channel state feedback,'' \emph{IEEE Trans. Inf. Theory}, vol. 56, no. 6, pp. 2845-2866, Jun. 2010.

\bibitem{achieving_massive}
H. Hoon, G. Caire, H. C. Papadopoulos and S. A. Ramprashad, ``Achieving ``massive MIMO" spectral efficiency with a not-so-large number of antennas,'' \emph{IEEE Trans. Wireless Commun.}, vol. 11, no. 9, pp. 3226-3239, Sep. 2012.

\bibitem{sparsity}
J. Shen, J. Zhang, K. Chen, and K. B. Letaief, ``High-dimensional CSI acquisition in massive MIMO: Sparsity-inspired approaches,'' \emph{IEEE Systems Journal}, to appear.

\bibitem{sparsity2}
J. Shen, J. Zhang, E. Alsusa, and K. B. Letaief, ``Compressed CSI acquisition in FDD massive MIMO: How much training is needed?'' \emph{IEEE Trans. Wireless Commun.}, to appear.

\bibitem{jointspatial}
A. Adhikary, J. Nam, J. Ahn, and G. Caire, ``Joint spatial division and multiplexing---The large-scale array regime," \emph{IEEE Trans. Inf. Theory}, vol.59, no.10, pp.6441-6463, Oct. 2013.

\bibitem{how_much}
B. Hassibi and  B. M. Hochwald, ``How much training is needed in multiple-antenna wireless links?,'' \emph{IEEE Trans. Inf. Theory}, vol. 49, no. 4, pp. 951-963, Apr. 2003.

\bibitem{optimal_channel}
J. Hoydis, M. Kobayashi, M. Debbah, ``Optimal channel training in uplink network MIMO systems,'' \emph{IEEE Trans. Signal Process.}, vol.59, no.6, pp. 2824-2833, Jun. 2011.

\bibitem{jakes}
J. G. Proakis, \emph{Digital communication}, 4th ed ed. New York: Mc-Graw-Hill, 2000.

\bibitem{predictor}
J. Isotalo, and S. Puntanen, ``Linear prediction sufficiency for new observations in the general Gauss-Markov model,'' \emph{Communications in Statistics-Theory and Methods}, vol. 35, no. 6, pp. 1011-1023, 2006.

\bibitem{equivalent}
W. Hachem, P. Loubaton, and J. Najim, ``Deterministic equivalents for certain functionals of large random matrices,'' \emph{The Annals of Applied Probability.}, vol. 17, no. 3, pp. 875-930, May. 2007.

\end{thebibliography}

\end{document}